\documentclass[conference]{IEEEtran}
\usepackage{amsmath,amssymb,amsthm}
\usepackage{algorithm,algorithmic}
\usepackage{setspace,flushend}
\algsetup{linenosize=\small}
\usepackage[left=0.65in,right=0.65in,
top=0.75in,bottom=1.02in]{geometry}
\newcommand{\cA}{\mathcal{A}}
\newcommand{\cB}{\mathcal{B}}
\newcommand{\cE}{\mathcal{E}}
\newcommand{\cF}{\mathcal{F}}
\newcommand{\cG}{\mathcal{G}}
\newcommand{\cH}{\mathcal{H}}
\newcommand{\cK}{\mathcal{K}}
\newcommand{\cI}{\mathcal{I}}
\newcommand{\cL}{\mathcal{N}}
\newcommand{\cM}{\mathcal{M}}
\newcommand{\cP}{\mathcal{P}}
\newcommand{\cT}{\mathcal{T}}
\newcommand{\cU}{\mathcal{U}}
\newcommand{\cV}{\mathcal{V}}
\newcommand{\cW}{\mathcal{W}}
\newcommand{\cX}{\mathcal{X}}
\newcommand{\cY}{\mathcal{Y}}

\newcommand{\bF}{\mathbb{F}}

\newcommand{\vxi}{\underline{\xi}}
\newcommand{\vdel}{\underline{\delta}}

\newcommand{\vv}{\mathbf{v}}
\newcommand{\vy}{\mathbf{y}}
\newcommand{\vz}{\mathbf{z}}
\newcommand{\mA}{\mathbf{A}}

\newcommand{\mM}{\mathbf{M}}
\newcommand{\mH}{\mathbf{H}}
\newcommand{\mR}{\mathbf{R}}
\newcommand{\mT}{\mathbf{T}}
\newcommand{\mV}{\mathbf{V}}
\newtheorem{thm}{Theorem}[section]

\begin{document} \sloppy

\title{Linear Network Coding for Multiple Groupcast Sessions: An Interference Alignment Approach}
\author{Abhik~Kumar~Das, Siddhartha~Banerjee and Sriram~Vishwanath\\
Dept. of ECE, The University of Texas at Austin, TX, USA\\
Email: \{akdas,siddhartha\}@utexas.edu, sriram@austin.utexas.edu}

\maketitle

\begin{abstract}
We consider the problem of linear network coding over communication networks, representable by directed acyclic graphs, with multiple groupcast sessions: the network comprises of multiple destination nodes, each desiring messages from multiple sources. We adopt an interference alignment perspective, providing new insights into designing practical network coding schemes as well as the impact of network topology on the complexity of the alignment scheme. In particular, we show that under certain (polynomial-time checkable) constraints on networks with $K$ sources, it is possible to achieve a rate of $1/(L+d+1)$ per source using linear network coding coupled with interference alignment, where each destination receives messages from $L$ sources $(L< K)$, and $d$ is a parameter, solely dependent on the network topology, that satisfies $0\leq d< K-L$.

\end{abstract}


\section{Introduction}
The problem of characterizing the capacity of communication networks and designing coding strategies with achievable rates close to network capacity has been an important topic of research. Ever since the development of the novel concept of \emph{linear network coding} (LNC) and its success in demonstrating the achievability of maximum throughput in multicast networks, extensions of the concept have been applied to more general settings for obtaining useful answers to network capacity problems \cite{Ahlswede_Yeung,Li2003,kotter,Jaggi2005}. However, both scalar and vector versions of LNC have been shown to be inadequate in typifying the limits of inter-session network coding \cite{lehman2,medard-nm,Dougherty}; and this has hindered the progress towards the development of coding schemes that provide improved rates or even guarantees on the achievable rates with respect to network capacity.

In this paper, we consider the problem of LNC for networks, represented by directed acyclic graphs, with multiple groupcast sessions. As defined in \cite{MalekiJafar2012,Syed2012}, a groupcast session refers to the setup where a destination is interested in messages from multiple (not necessarily all) sources, or analogously, messages from a source are transmitted to multiple (not necessarily all) destinations. Thus, unicast (one source to one destination) and broadcast (one source to all destinations) can be thought of as special cases of groupcast. The use of LNC results in a linear transfer function representation for the network in terms of its transmission streams; these streams can ``mix" with each other and generate ``interference" at the destinations  \cite{kotter,KK}. A sufficient but somewhat restrictive condition for interference-free transmission in networks employing LNC is derived in \cite{kotter}, but it is highly non-trivial to design LNC schemes that satisfy this condition in case of multiple transmission sessions.

We analyze the problem of LNC over networks with multiple groupcast sessions from an interference alignment perspective, along the lines of \cite{das10,Ramakrishnan2010,Han2011,meng12}, that look at the problem of LNC for networks with three unicast sessions. The motivation for this approach comes from the fact that a network with groupcast sessions and using LNC is analogous to a generalized version of the information-theoretic interference channel where each destination desires messages from multiple sources. A similar approach has been adopted in the context of analyzing multiple groupcasts for index coding \cite{MalekiJafar2012,Syed2012}. We focus on designing coding schemes based on LNC coupled with interference alignment, through the use of precoding matrices and multiple transmissions \cite{Cadambe-Jafar}. We also examine the effect of groupcast configurations and network topology on achievable source rates and the ease of using alignment methods for decoding relevant messages at destinations.

\textbf{Related Work:} The problem of designing inter-session LNC schemes achieving specific source rates for general network topologies has been shown to be NP-hard \cite{lehman2}; this has prompted the development of sub-optimal constructive LNC schemes for networks with multiple transmission sessions; examples include packing a network using poison-antidote butterflies \cite{poison1}, linear programming \cite{poison2}, and network tiling for networks based on triangular lattices \cite{tiling}. The problem of determining the feasibility of LNC and constructing coding schemes for two multicast sessions is analyzed in \cite{WangShroff07,SongCai11} using a graph-theoretic approach. The use of interference alignment methods alongside LNC has been examined in \cite{das10,Ramakrishnan2010,meng12} for three unicast sessions, where the main result is that each source can achieve a rate close to half the mincut using large enough number of transmissions, if the mincuts are ones and certain network constraints are satisfied.

\textbf{Main Results:} In this paper, we introduce the concept of an {\em interference graph} associated with a network having multiple groupcast sessions. We utilize this interference graph to design precoding matrices, and examine its impact on achievable source rates. In particular, for networks with $K$ sources and mincuts of either zero or one for any source-destination pair, we show that, if the interference graph has no cycles, then each source can achieve a rate of $\frac{1}{L+1}$ using LNC coupled with interference alignment over $(L+1)$ transmissions, given that each destination is interested in messages from $L$ ($L<K$) sources and a set of polynomial-time checkable network constraints are satisfied. We obtain a weaker achievability result if the interference graph has cycles -- we show that a rate of $\frac{1}{L+d+1}$ per source can be achieved with interference alignment over $(L+d+1)$ transmissions under certain polynomial-time checkable network constraints, where $d$ depends only on the topology of interference graph and satisfies $0\leq d< K-L$. Furthermore, we develop an algorithm that gives the optimal (or smallest feasible) value of $d$ for a given interference graph, and hence the maximal rates for the coding scheme.

The rest of the paper is organized as follows. We describe the system model and preliminaries in Section \ref{sec:sysmodel}. We design coding schemes and state the achievability results in Section \ref{sec:results}. Finally, we conclude the paper with Section \ref{sec:conclusion}.

\section{System Model and Preliminaries}\label{sec:sysmodel}
{\em Notation:} We use $\bF_q$ to represent the finite field with $q$ elements, where $q$ is a prime number or its power. Given a vector of variables $\underline{z}$, we use $\bF_q[\underline{z}]$ to denote the polynomial ring over $\bF_q$ constructed using variables in $\underline{z}$. For any matrix $\mA$ and vector space $\cU$ over some field, we use $\textrm{span}(\mA)$ to refer to the vector space spanned by the column vectors of $\mA$, and $\textrm{dim}(\cU)$ to represent the dimension number of $\cU$.

\textbf{System Model:} We consider a communication network represented by a directed acyclic graph $\cG=(\cV,\cE)$, where $\cV$ is the set of nodes and $\cE$ is the set of directed links. We assume that each link represents a noiseless channel and transmissions across different links do not interfere with each other. There are $K$ sources $S_1,S_2,\ldots,S_K$, and $M$ destinations $D_1,D_2,\ldots,D_M$, among the nodes in $\cV$. We have multiple groupcast sessions in $\cG$, i.e., $D_i$ is interested in messages from some subset of sources, say $\cA_i\subset\{S_1,S_2,\ldots,S_K\}$. For the sake of simplicity, we let $|\cA_i|=L$ for all $i$. We assume that the messages generated by different sources are probabilistically independent and transmitted in form of symbols from $\bF_q$, where $q$ is a prime or its power. We also restrict the capacity of links in $\cE$ to one symbol (from $\bF_q$) per transmission.

We employ LNC for communication between the sources and destinations in $\cG$. In other words, every node generates and transmits linear combinations of its received packets, where the coefficients for linear combination come from $\bF_q$. These coefficients can be treated as variables, say $\xi_1,\xi_2,\ldots,\xi_s$, that take values from $\bF_q$. Then a LNC scheme refers to choosing a suitable assignment of $\vxi:=[\xi_1\,\,\xi_2 \,\cdots \,\xi_s]$ from $\bF_q^s$.

The generalization of \emph{Max-flow Mincut} Theorem for networks states that the transmission rate between a source and destination is bounded above by the mincut between them \cite{kotter}. As a starting point for tackling the problem of designing codes for groupcast sessions, we assume that the mincut between $S_j$ and $D_i$ is one if $S_j\in\cA_i$, and at most one for remaining choices of $i,j$ -- this ensures that $D_i$ is connected to sources in $\cA_i$ and it can receive at most one symbol per transmission from them. We define $x_j\in\bF_q$ as the symbol transmitted by $S_j$ and $m_{ij}(\vxi)\in\bF_q[\vxi]$ as the transfer function between $S_j$ and $D_i$  for some transmission. Then, the symbol received by $D_i$, also denoted by $y_i\in\bF_q[\vxi]$, is given by the following relation:
\[
y_i = \sum_{j=1}^K m_{ij}(\vxi)x_j,\quad i = 1,2,\ldots,M.
\]
Since $D_i$ is only interested in messages from sources in $\cA_i$, the presence of non-zero transfer functions $m_{ij}(\vxi)$, $S_j\not\in \cA_i$, acts as ``interference" to the decoding processes at the destinations. Note that $m_{ij}(\vxi)\not\equiv 0$ for $S_j\in\cA_i$, since the mincut between each source in $\cA_i$ and $D_i$ is one. Also, the mincut between $S_j$ and $D_i$ being zero for some $i,j$ implies that $m_{ij}(\vxi)\equiv 0$. We define $\cB_i=\{S_j\not\in\cA_i:m_{ij}(\vxi)\not\equiv 0\}$ -- the set of interfering sources for $D_i$. We also assume $\cB_i\neq \emptyset$ for all $i$ -- this ensures the presence of ``interference" at each of the destinations.

\textbf{Interference Alignment:} The presence of interfering transfer functions in the system model described above, motivates the need for interference alignment in conjunction with LNC schemes \cite{das10,Ramakrishnan2010,Han2011,meng12}. In particular, we focus on the application of alignment schemes based on symbol-extension that ensures the sources in $\cG$ are able to transmit at equal rates. We consider $n$ time-slots or transmissions and define $\vxi^{(k)}$ as the assignment of $\vxi$ for the $k$th transmission, $k=1,2,\ldots,n$. Given $a,b,n$ such that $a\leq b$ and $n\geq La+b$, we define $\vz_i\in\bF_q^{a\times 1}$ as the message vector of $S_i$, and consider a $n\times a$ precoding matrix $\mV_i$ that encodes $\vz_i$ into $n$ symbols. Then $D_i$ receives a $n\times 1$ vector $\vy_i$ that satisfies the following relation:
\[
\vy_i=\sum_{j=1}^K\mM_{ij}\mV_j\vz_j,\quad i=1,2,\ldots,M.
\]
Note that $\mM_{ij}$ is a $n\times n$ diagonal matrix with $m_{ij}(\vxi^{(k)})$ as the $(k,k)$th entry. We define $\vdel$ as the vector of variables in $\vxi^{(1)},\vxi^{(2)},\ldots,\vxi^{(n)}$ and those used for constructing the precoding matrices, and also the following vector spaces over $\bF_q[\vdel]$:
\begin{align*}
\cU_i&=\textrm{span}([\mM_{ij}\mV_j:\,\,S_j\in\cA_i]),\\
\cW_i&=\textrm{span}([\mM_{ij}\mV_j:\,\,S_j\in\cB_i]),
\end{align*}
for $i=1,2,\ldots,M$. Then the interference alignment approach seeks to design precoding matrices that satisfy the following conditions for some assignment of the entries of $\vdel$ from $\bF_q$:
\begin{equation}\label{eqn:alignment}
\textrm{dim}(\cU_i)=La,\,\,\textrm{dim}(\cW_i)=b,\,\,\textrm{dim}(\cU_i\cap\cW_i)=0,
\end{equation}
for $i=1,2,\ldots,M$. The constraint on the dimension of $\cW_i$ maps the interference vectors to a single subspace at each destination. The constraint on the dimension of $\cU_i\cap\cW_i$ guarantees that the subspace spanned by the interference vectors is linearly independent of the subspace spanned by the desired vectors; this along with the constraint on the dimension of $\cU_i$ permits error-free recovery of the desired messages. Therefore, $S_i$ can transmit $a$ symbols in $n$ transmissions, thereby achieving a rate of $\frac{a}{n}$ -- we refer to this network coding scheme as \emph{precoding-based network alignment} (PBNA), along lines of \cite{meng12}.

\textbf{Interference Graph:} We define an undirected bipartite graph $\cH=(\cX,\cY,\cF)$, where $\cX=\{S_1,S_2,\ldots,S_K\}$, $\cY=\{\cW_1,\cW_2,\ldots,\cW_M\}$ are the node partitions, and $\cF$ is the set of undirected edges such that $(S_j,\cW_i)\in\cF$ if and only if $S_j\in\cB_i$. Thus, $\cH$ encodes the set of sources whose signals act as interference, and therefore, need to be aligned/mapped to a single subspace at each destination - hence, we refer to it as the \emph{interference graph}. Note that the topology of $\cH$ has a direct bearing on the achievable rates of the sources; for example, abundant low-degree nodes in $\cY$ and smaller values of $|\cF|$ could result in potentially higher achievable rates due to lesser number of interference terms (and therefore, alignment constraints) at the destinations. We explore this connection, using PBNA as the coding scheme, in the next section.


\section{Results for Achievable Rates}\label{sec:results}
We first consider the case where the interference graph $\cH$ has no cycles. Then we have the following achievability result:
\begin{thm}\label{thm:nocycle}
If $\cH$ has no cycles, then one can achieve a rate of $\frac{1}{L+1}$ per source using PBNA, provided the finite field size $q$ is chosen to be sufficiently large and certain constraints (checkable in time that is polynomial in $L,|\cF|$ and transfer function degrees) are satisfied by the transfer functions.
\end{thm}
\begin{proof}
We prove this achievability result by setting $n=(L+1)$, $a=b=1$, and designing precoding matrices $\mV_i$, $i=1,2,\ldots,K$, that satisfy the relations in \eqref{eqn:alignment}. Note that since $\cH$ has no cycles, it is either a tree or a collection of disjoint trees. We assume there are $c\geq 1$ disjoint trees and denote them by $\cT_l=(\cX_l,\cY_l,\cF_l)$, $l=1,2,\ldots,c$. Thus, $\{\cX_l:\,l=1,2,\ldots,c\}$, $\{\cY_l:\,l=1,2,\ldots,c\}$, $\{\cF_l:\,l=1,2,\ldots,c\}$ are partitions of $\cX,\cY,\cF$ respectively, and $\cH=\cup_{l=1}^c\cT_l$. Note that if $\cH$ is a single spanning tree, then we have $c=1$ and $\cH=\cT_1$.

We handle the disjoint trees separately, i.e., the precoding vectors for sources in $\cX_l$ are designed independently of those for sources in $\cX_k$, $k\neq l$. Given $l\in\{1,2,\ldots,c\}$, we choose any $S_{a_l}\in\cX_l$ as the tree root. Next, we define $\cL_0^{(l)}=\{S_{a_l}\}$, $\cL_1^{(l)}$ as the set of neighbor nodes of $S_{a_l}$ in $\cT_l$, and $\cL_{k+1}^{(l)}$ as the set of neighbors of nodes in $\cL_k^{(l)}$ for $k\geq 1$ (these sets are levels of the BFS tree rooted at $S_{a_l}$). Since $\cT_l$ is a bipartite graph, $\cL_{2k+1}^{(l)}\subseteq \cY_l$ and $\cL_{2k}^{(l)}\subseteq \cX_l$ for $k\geq 0$. Thereafter, we associate a $(L+1)\times(L+1)$ matrix $\mH_{ij}$ with $(S_j,\cW_i)\in\cF_l$:
\[
\mH_{ij}:=
\begin{cases}
\mM_{ij}, & S_j\in\cL_{2k}^{(l)},\,\,\cW_i\in\cL_{2k+1}^{(l)},\,\,k\geq 0,\\
\mM_{ij}^{-1}, & S_j\in\cL_{2k+2}^{(l)},\,\,\cW_i\in\cL_{2k+1}^{(l)},\,\,k\geq 0.
\end{cases}
\]
Thus, by construction, $\mH_{ij}$ is a diagonal matrix with $h_{ij}(\vxi^{(k)})$ as $(k,k)$th entry, such that $h_{ij}(\vxi)\equiv m_{ij}(\vxi)$ for $S_j\in\cL_{2k}^{(l)}$, $\cW_i\in\cL_{2k+1}^{(l)}$, $(S_j,\cW_i)\in\cF_l$, and $h_{ij}(\vxi)\equiv (m_{ij}(\vxi))^{-1}$ for $S_j\in\cL_{2k+2}^{(l)}$, $\cW_i\in\cL_{2k+1}^{(l)}$, $(S_j,\cW_i)\in\cF_l$, for $k\geq 0$.

We set $\mV_{a_l}=[\theta^{(1)}_l\,\,\theta^{(2)}_l\,\cdots\,\theta^{(L+1)}_l]^T$, where $\theta^{(k)}_l$, $k=1,2,\ldots,L+1$, are variables drawing values from $\bF_q$. Since $\cH$ is a collection of trees, there exists a unique path between $S_u$ and $S_v$, say $\cP_{uv}$, if they are connected to each other via edges. Given $i\neq a_l$ and $S_i\in\cX_l$, we set $\mV_i=\mT_i\mV_{a_l}$, where
\[
\mT_i=\prod_{(u,v):(S_v,\cW_u)\in\cP_{i,a_l}}\mH_{uv}.
\]
$\mT_i$ is a diagonal matrix with $(k,k)$th entry as $t_i(\vxi^{(k)})$, where
\[
t_i(\vxi)\equiv \prod_{(u,v):(S_v,\cW_u)\in\cP_{i,a_l}}h_{uv}(\vxi).
\]
This choice of precoding vectors ensures $\mM_{ij}\mV_j=\mM_{ik}\mV_k$ if $S_j,S_k\in\cB_i$ and $|\cB_i|\geq 2$. Therefore, the constraint $\textrm{dim}(\cW_i)=1$ is satisfied for all $i$ (this is trivially satisfied if $|\cB_i|=1$). Also, the constraints $\textrm{dim}(\cW_i)=1$, $\textrm{dim}(\cU_i\cap\cW_i)=0$ are satisfied if and only if the set of vectors $\{\mM_{ij}\mV_j:\,S_j\in\cA_i\}$ and $\mM_{ik}\mV_k$, for any $k\in\cB_i$, form a full rank $(L+1)\times (L+1)$ matrix, say $\mR_{ik}$. Note that the entries of $\mR_{ik}$ are rational functions based on polynomials in $\bF_q[\vdel]$, where $\vdel$ comprises of variables in $\vxi^{(k)}$ and $\theta^{(k)}_l$, $k=1,2,\ldots,L+1$, $l=1,2,\ldots,c$.

The fact whether $\mR_{ik}$ is full rank or not can be checked by computing the determinant of $\mR_{ik}$ -- if the determinant is a rational function with non-zero numerator-denominator product, say $r_{ik}(\vdel)\in\bF_q[\vdel]$, then $\mR_{ik}$ is full rank, else it is not. Also, computing these determinant values require time that is polynomial in $L,|\cF|$ and the transfer function degrees. Therefore, we need the following polynomial to be non-trivial:
\[
f(\vdel)=\prod_{k=1}^{L+1}\prod_{(i,j):m_{ij}(\vxi)\not\equiv 0}m_{ij}(\vxi^{(k)})\prod_{i=1}^K\prod_{k\not\in\cA_i}r_{ik}(\vdel),
\]
for satisfying all constraints in \eqref{eqn:alignment}. An assignment of $\vdel$ from $\bF_q^{(L+1)(s+c)}$ that makes $f(\vdel)$ non-zero is guaranteed for large enough finite field size $q$ using a simplified version of Schwartz-Zippel Lemma \cite{das10,meng12}. Therefore, this assignment of $\vdel$ enables each source to transmit at rate of $\frac{1}{L+1}$.
\end{proof}

Thus, the absence of cycles in the interference graph enables one to choose a set of precoding matrices/vectors independently of each other (corresponding to the sources that are chosen as roots of the disjoint trees in the above proof) and use them to construct precoding matrices/vectors for the remaining sources. Moreover, PBNA makes use of exactly $(L+1)$ transmissions to enable each source to transmit one message, thereby achieving a total sum rate of $\frac{K}{L+1}$.

The presence of cycles in the interference graph can impose restrictions on the precoding matrices that may the affect the ease of using alignment schemes. We illustrate this using a $4\times 4$ network -- we set $\cA_1=\{S_3,S_4\}$, $\cA_2=\{S_1,S_4\}$, $\cA_3=\{S_1,S_2\}$, $\cA_4=\{S_2,S_3\}$, and assume all the transfer functions are non-trivial. We also define the following rational function:
\[
t(\vxi)\equiv \frac{m_{12}(\vxi)m_{23}(\vxi)m_{34}(\vxi)m_{41}(\vxi)}
{m_{11}(\vxi)m_{22}(\vxi)m_{33}(\vxi)m_{44}(\vxi)}.
\]
It is easy to see that the resulting interference graph $\cH$ is a cycle;
we now have the following negative result for this setup:
\begin{thm}\label{thm:cycle}
If $\cH$ is the cycle interference graph of the network described above and $t(\vxi)$ is a non-constant rational function (i.e., $t(\vxi)\not\equiv a$, $a\in\bF_q$), then one cannot achieve a rate of $\frac{1}{3}$ per source in finite number of transmissions.
\end{thm}

\begin{proof}
Note that $L=2$ for this case, therefore, we require $a=b$ and $n=2a+b$ to achieve a rate of $\frac{a}{n}=\frac{1}{3}$ per source. Then the constraints $\textrm{dim}(\cW_i)=a$ for all $i$, as given in \eqref{eqn:alignment}, imply that precoding matrices satisfy the following relations:
\begin{align*}
\mM_{11}\mV_1\mA_1 &= \mM_{12}\mV_2, \quad\mM_{22}\mV_2\mA_2 = \mM_{23}\mV_3,\\
\mM_{33}\mV_3\mA_3 &= \mM_{34}\mV_4, \quad\mM_{44}\mV_4\mA_4 = \mM_{41}\mV_1,
\end{align*}
where $\mA_i$, $i=1,2,3,4$, are full rank $a\times a$ matrices. These relations result in the equation $\mT\mV_1=\mV_1\mA$, where $\mT=\mM_{12}\mM_{23}\mM_{34}\mM_{41}(\mM_{11}\mM_{22}\mM_{33}\mM_{44})^{-1}$ and $\mA = \mA_1\mA_2\mA_3\mA_4$; this imposes restrictions on choices of $\mV_1$. $\mT$ is a diagonal matrix with its $(k,k)$th entry as $t(\vxi^{(k)})$. Thus, we have $t(\vxi^{(k)})\vv_k = \vv_k\mA$, where $\vv_k$ is the $k$th row of $\mV_1$, $k=1,2,\ldots,n$. This means if $\vv_k$ is not the zero vector for some $k$, then it is one of the left eigenvectors of $\mA$ and $t(\vxi^{k})$ is the corresponding eigenvalue \cite{Foote}. Since $t(\vxi)$ is not a constant, the eigenvectors form a linearly independent set and $\mA$ is full rank, $\vv_k$ is the zero vector for $(n-a)$ instances of $k$, i.e., $(n-a)$ rows of $\mV_1$ are zero vectors. Then the four alignment relations stated above imply that the corresponding $(n-a)$ rows of $\mV_i$, $i=2,3,4$, are also zero vectors. One can check that these precoding matrices satisfy $\textrm{dim}(\cU_i\cap\cW_i)>0$ for all $i$, that makes recovery of desired messages impossible at each destination. Therefore, the sources cannot achieve a rate of $\frac{1}{3}$ each with $a= b$, using PBNA. However, if $a<b$ and the relations in \eqref{eqn:alignment} could be satisfied by some choice of precoding matrices, the achievable rate per source would be at most $\frac{a}{2a+b}=\frac{1}{2+(b/a)}$. Hence, the only possibility for achieving rate close to $\frac{1}{3}$ per source is to choose $a,b$ large enough such that $\frac{b}{a}$ is very close to one; this in turn introduces the requirement that the number of transmissions $n$ should be large.
\end{proof}

Thus, the presence of cycles in the interference graph can result in PBNA requiring large number of transmissions for each source to achieve a rate close to $\frac{1}{L+1}$ and sum rate close to $\frac{K}{L+1}$. The reason for this, as observed in the network example above, are restrictions arising on the choice of precoding matrices. One way of tackling this problem is to allow the destinations to decode some of the interference messages, i.e., $D_i$ agrees to decode messages from some sources in $\cB_i$ along with those from sources in $\cA_i$. This approach reduces the number of relations in \eqref{eqn:alignment} to be satisfied, thereby effectively removing edges from the interference graph $\cH$. For example, if $D_i$ decodes messages from $S_j\in\cB_i$ ($\cB_i\neq \emptyset$), the alignment constraints involving $\mV_j$ that need to be satisfied at $D_i$ get eliminated; this is equivalent to removing $(S_j,\cW_i)\in\cF$ from $\cH$. However, the tradeoff of this approach is decrease in the transmission rate of sources since each destination has to recover messages from potentially more than $L$ sources.

We define $\cE_i\subseteq \cB_i$ as the set of \emph{extra} sources whose messages are decoded by $D_i$, so that $D_i$ now recovers messages from sources in $\bar{\cA}_i=\cA_i\cup\cE_i$, and the new interfering set of sources for $D_i$ is $\bar{\cB}_i=\cB_i\setminus\cE_i$. This corresponds to removing the set of edges $\{(S_j,\cW_i):S_j\in\cE_i\}$ from $\cH$ to get a new interference graph $\bar{\cH}=(\cX,\cY,\bar{\cF})$, where $\bar{\cF}=\{(S_j,\cW_i):S_j\in\bar{\cB}_i\}$. Our objective is to remove cycles in $\bar{\cH}$ -- thereafter, we can use PBNA to achieve certain source rates. In particular, we have the following achievability result:
\begin{thm}
Suppose $\bar{\cH}$, generated from an interference graph $\cH$ as described above, has no cycles, and let $d=\max\{|\cE_i|:\,i=1,2,\ldots,M\}$. Then one can achieve a rate of $\frac{1}{L+d+1}$ per source using PBNA, provided the finite field size $q$ is chosen to be sufficiently large and certain constraints (checkable in time that is polynomial in $L$, $d$, $|\bar{\cF}|$ and transfer function degrees) are satisfied by the transfer functions.
\end{thm}
\begin{proof}
Note that $|\bar{\cA}_i|=L+|\cE_i|\leq L+\min(d,|\cB_i|)$, and $|\bar{\cA}_i|=L+d$ for at least one $D_i$. If $|\bar{\cA}_i|=L+d$ and $\bar{\cB}_i\neq\emptyset$ for all $i$, since $\bar{\cH}$ has no cycles, we can directly apply Theorem \ref{thm:nocycle} to achieve a rate of $\frac{1}{L+d+1}$ per source using PBNA under certain constraints. The only other case is that $|\bar{\cA}_i|<L+d$ and/or $\bar{\cB}_i=\emptyset$ for some values of $i$ ($\bar{\cB}_i=\emptyset$ implies $D_i$ chooses to decode messages from all sources in $\cB_i$). Then we can introduce unique artificial transfer functions and auxiliary sources to make $|\bar{\cA}_i|=L+d$ and $\bar{\cB}_i\neq \emptyset$ for each such $i$. For example, if $|\bar{\cA}_i|=L+d-2$ and $\bar{\cB}_i=\emptyset$ for some $i$, we construct three dummy sources, say $S_1',S_2',S_3'$, and assume that the corresponding transfer functions with respect to $D_i$ are variables $\eta_{i1}$, $\eta_{i2}$ and $\gamma_{i3}$ respectively (that take values from $\bF_q$). Thereafter, we make the updates $\bar{\cA}_i\leftarrow \bar{\cA}_i \cup\{S_1',S_2'\}$ and $\bar{\cB}_i\leftarrow \{S_3'\}$, so that $|\bar{\cA}_i|=L+d$ and $|\bar{\cB}_i|=1$. Thus, this procedure ensures $|\bar{\cA}_i|=L+d,\,\bar{\cB}_i\neq\emptyset$ for all $i$, and we can use Theorem \ref{thm:nocycle} to complete the achievability proof.
\end{proof}

If $\cH$ has cycles, there can be multiple candidates for subgraph $\bar{\cH}$ that has no cycles. Since we want to maximize the transmission rates for the sources, we are interested in the smallest value that $d$ can take -- we refer to this as $d^*$. Therefore, we have the following graph-theoretic optimization problem over $\cH$ -- what is the minimum value of $d$ so that if we remove some set of $\min(d,|\cB_i|)$ edges from node $\cW_i\in\cY$ ($|\cB_i|$ is the degree of node $\cW_i$), the resulting graph $\bar{\cH}$ has no cycles? We first assume that $\cH$ is a connected graph. Then a modified optimization problem, that gives the same optimal value $d^*$, is -- what is the minimum value of $d$ so that if we remove at most $d$ edges from each node in $\cY$, the resulting graph $\cK$ is a spanning tree of $\cH$? We denote the optimal $\bar{\cH}$ and $\cK$, obtained as solutions to these optimization problems, by $\bar{\cH}^*$ and $\cK^*$ respectively. Note that $\bar{\cH}^*$ can obtained from $\cK^*$ by removing edges from $\cK^*$, if needed, such that the difference between degrees of $\cW_i$ in $\cH$ and $\bar{\cH}^*$ is $\min(d^*,|\cB_i|)$.

To obtain an algorithm for this, we make use of the concepts from matroid theory; refer to \cite{CLRS} for details. Given that $\cH$ is connected, we consider its graphic matroid $\cM$, i.e., the collection of all acyclic edge-sets of $\cH$ (i.e. tree/forest subgraphs of $\cH$) -- the bases (maximal elements) of this matroid are the spanning trees of $\cH$. The dual of $\cM$, denoted by $\bar{\cM}$, is defined as the collection of all edge-sets of $\cH$ whose complement graph contains a spanning tree of $\cH$ (or, is connected) - from matroid theory, we have that this is also a matroid (often called the co-graphic or bond matroid). Finally, given $d$, we have the partition matroid $\cM_d$ composed of all edge-sets of $\cH$ such that at most $d$ edges are chosen from each node in $\cY$. As the matroid property is closed under intersections, $\bar{\cM}\cap \cM_d$ is also a matroid. We define an (arbitrary) labeling of edges of $\cH$ as $\cF=\{e_1,e_2,\ldots,e_{|\cF|}\}$. Using these definitions, we make use of Algorithm \ref{alg:greedy} to obtain $d^*$ and $\cK^*$ for connected $\cH$.

\begin{thm}
Given any (connected) bipartite graph $\cH$, Algorithm \ref{alg:greedy} finds $d^*$ and the optimal spanning tree $\cK^*$ for $\cH$, with a computational complexity of $O(\left(1+\frac{K}{M}\right)|\cF|^2)$.
\end{thm}
\begin{proof}
Note that $d^*$ lies between $\lceil(|\cF|-K-M+1)/M\rceil$ and $\lfloor|\cF|/M\rfloor$ since $\cH$ is connected, number of edges in its spanning tree is $K+M-1$, and at most $Md$ edges are removed from $\cH$ to obtain $\cK^*$. Also, the edge-set of the complement of $\cK^*$ is a maximal independent set of $\bar{\cM}\cap\cM_{d^*}$. The inner loop of the algorithm corresponds to the greedy approach for generating a maximal independent set of $\bar{\cM}\cap\cM_d$ for given $d$. The outer loop of the algorithm checks if the complement of the maximal independent set forms a spanning tree for $\cH$ by examining if the size of the obtained maximal independent set is $|\cF|-K-M+1$ or not. This, along with the fact that all maximal independent sets of a matroid have the same size, shows that the algorithm returns $d^*$, $\cK^*$ as answers.

The membership of $\cI\cup\{e_i\}$ in $\bar{\cM}\cup\cM_d$ for each $i$ and $d$ can be checked using BFS (or DFS) algorithm in $O(|\cF|)$ time. The inner for-loop runs $|\cF|$ times and the outer for-loop runs at most $\left(1+\frac{K}{M}\right)$ times. Therefore, the algorithm has an
overall computational complexity of $O(\left(1+\frac{K}{M}\right)|\cF|^2)$.
\end{proof}

\setlength{\textfloatsep}{0pt}
\begin{algorithm}[t]
{ \fontsize{10}{10}
\caption{Finding $d^*$ and $\cK^*$ for connected $\cH$}\label{alg:greedy}
\begin{algorithmic}
\REQUIRE $\cH=(\cX,\cY,\cF)$, $\cF=\{e_1,e_2,\ldots,e_{|\cF|}\}$
\FOR{$d=\lceil(|\cF|-K-M+1)/M\rceil$ to $\lfloor|\cF|/M\rfloor$}
    \STATE $\cI\leftarrow \emptyset$
    \FOR{$i=1$ to $|\cF|$}
        \IF{$\cI\cup\{e_i\}\in\bar{\cM}\cap\cM_d$}
            \STATE $\cI\leftarrow\cI\cup\{e_i\}$
        \ENDIF
    \ENDFOR
    \IF{$|\cI|=|\cF|-K-M+1$}
        \STATE \textbf{break}
    \ENDIF
\ENDFOR
\RETURN{$d^*=d$, $\cK^*=(\cX,\cY,\cF\backslash \cI)$}
\end{algorithmic}}
\end{algorithm}

In case $\cH$ is not a connected component, we can apply Algorithm \ref{alg:greedy} on its disjoint components separately and obtain their corresponding optimal values of $d$ and optimal spanning trees. Then $d^*$ is the maximum of the optimal values of $d$ obtained for the disjoint components, and $\bar{\cH}^*$ can be obtained using an edge removal process from the set of disjoint optimal trees similar to the one used for the case of connected $\cH$. If number of components of $\cH$ is $c$, the time complexity for running Algorithm \ref{alg:greedy} over them is $O(c\left(1+\frac{K}{M}\right)|\cF|^2)$.

\section{Conclusion}\label{sec:conclusion}
The main goal of this paper is to provide a systematic framework for presenting guarantees on achievable rates for networks employing LNC with multiple groupcast sessions. We use PBNA for designing codebooks based on finite number of transmissions for networks with acyclic interference graphs. For networks with cyclic interference graphs, we show this may not be possible, and instead, present a graph sparsification scheme with bounds on the resulting achievable source rates. Some of the future directions related to this problem include designing coding strategies that give higher throughput guarantees and generalization to arbitrary mincut values.

\bibliographystyle{ieeetr}
\bibliography{refa,refb,refc,refd}
\end{document}